\documentclass[a4paper,UKenglish,cleveref,autoref,pdfa]{lipics-v2021}
\usepackage{graphicx} 
\usepackage[full,disableredefinitions,small]{complexity}
\usepackage[size=footnotesize, color=blue!20]{todonotes}
\usepackage{xspace}
\usepackage{tcolorbox}
\usepackage{cite}
\usepackage{optidef} 
\usepackage{xpatch}

\usepackage{booktabs}
\usepackage{multirow}

\usepackage[size=footnotesize, color=blue!20]{todonotes}

\crefname{claim}{claim}{claims}
\Crefname{claim}{Claim}{Claims}

\definecolor{light blue}{rgb}{0.23,0.62,0.81}
\let\emph\relax\DeclareTextFontCommand{\emph}{\color{light blue}\em}

\DeclareMathOperator{\costone}{cost_1}
\DeclareMathOperator{\costinf}{cost_\infty}
\DeclareMathOperator{\OPTspread}{OPT_{\mathrm{spread}}}
\DeclareMathOperator{\OPTconn}{OPT_{\mathrm{conn}}}

\newcommand{\connectivity}{\textsc{Connectivity-Maintenance}\xspace}
\newcommand{\spreading}{\textsc{Points-Spreading}\xspace}

\title{Linear-Time Transformations Between Connectivity Maintenance and Points Spreading on Linear and Cyclic Domains}
\titlerunning{Transformations Between Connectivity Maintenance and Points Spreading}

\author{Nicolás Honorato-Droguett}
  {Nagoya University, Japan \and \url{https://nhdroguett.com}}
  {honorato.droguett.nicolas.n7@s.mail.nagoya-u.ac.jp}
  {https://orcid.org/0009-0005-1969-3649}
  {}
\authorrunning{N.~Honorato-Droguett}
\Copyright{Nicolás Honorato-Droguett}


\acknowledgements{I would like to thank Alexander Wolff for the fruitful discussions where the ideas of this paper first came to me.}

\nolinenumbers
\hideLIPIcs


\begin{document}

\maketitle

\begin{abstract}
    Given $n$ points on a line or closed cycle and a threshold $r>0$, the \emph{connectivity-maintenance} problem is to move the points so that every gap between consecutive points is at most $r$, whereas the \emph{points-spreading} problem requires every gap to be at least $r$.
    Li and Wang [CCCG 2015; CGT 2025] and Chen, Gu, Li, and Wang [SWAT 2012; DCG 2013] gave $O(n)$-time algorithms for the cyclic versions of min-max points-spreading and min-max connectivity-maintenance, respectively. Ghadiri and Yazdanbod [CCCG 2016] gave an $O(n\log n)$-time algorithm for the linear version of min-sum points-spreading.
    In this paper, we show that the two problems can be reduced in linear time to each other for both objectives and on both linear and cyclic domains.
    As an implication, min-sum connectivity-maintenance is solvable in $O(n\log n)$ time on both domains.
    Finally, we extend the reduction to points on a line with individual thresholds when their initial order is preserved.
\end{abstract}

\section{Introduction}

In the context of wireless sensor networks, the positions of sensors are crucial to maintaining coverage and connectivity.
When their initial positions do not satisfy the required constraints, moving the sensors is one way to improve coverage or connectivity~\cite{Bredin2005,Zavlanos2007,Wang2006}.
Even in one-dimensional domains, vehicles or sensors positioned along a road are moved to maintain network connectivity~\cite{Li2019A}, while sensors on a line or closed boundary are moved to establish barrier coverage~\cite{Czyzowicz2010,Andrews2016,Li2026}.
By contrast, sensors may be reallocated to reduce interference caused by overlapping sensing ranges~\cite{Kranakis2016,Kapelko2022}.

These tasks lead to two movement problems with opposite distance constraints.
Given $n$ points on a line or closed cycle and a positive threshold $r$, in {\connectivity} we want to move the points so that every gap between consecutive points in the final order is at most $r$, whereas in {\spreading} we wish for every such gap to be at least $r$.
On a closed cycle, distances and movements are measured along the closed cycle, and unlike the linear version, the constraint of {\connectivity} applies to every cyclic gap, including the gap between the last and first points~\cite{Chen2013,Li2026}.
For either problem, the objective is to minimise either the maximum movement of a point or the total movement, giving the min-max and min-sum objectives, respectively.

Although the two problems have opposite distance constraints, they otherwise share the same basic structure.
In both problems we are given the same set of points and threshold, allow the same movements, and use the same objective functions.
Interestingly enough, this similarity also appears in the known algorithms for the cyclic versions of min-max {\connectivity} and {\spreading}.
For {\connectivity}, the authors of~\cite{Chen2013} find a consecutive sequence of points whose gaps are too large.
Their algorithm moves the first and last points of this sequence towards each other, places the points between them at distance $r$, and then greedily moves the remaining points.
For {\spreading}, the greedy algorithm due to Li and Wang~\cite{Wang2025} first moves the points in one direction and then shifts all points back by half of the largest movement.
The recent approach by Li~\cite{Li2026} for min-max {\connectivity} follows the same approach, but instead reduces gaps larger than $r$.
In all three approaches, the optimal value is determined by a consecutive sequence of gaps whose total length is too large for {\connectivity} or too small for {\spreading}.
Despite this similarity, the two problems considered here have been formulated and studied separately.
For the min-sum objective, algorithms are known for {\spreading}~\cite{Ghadiri2016,HonoratoDroguett2026}.
To the best of our knowledge, no similar approach has been given for {\connectivity}.

Seeing these similarities, it is natural to ask whether these two problems are, in fact, closely related.
Our results show that they are.
We give transformations between the two problems in both directions, for both objectives and on both domains, such that the corresponding optimal costs differ only by a scaling factor.
Together with the known algorithms for min-sum {\spreading}, this implies that min-sum {\connectivity} is solvable in $O(n\log n)$ time on a line and a closed cycle.

\subparagraph{Related Work}
Both problems belong to the general study of minimising the movement of geometric objects subject to distance constraints~\cite{Demaine2009}.
For min-max {\spreading}, Dumitrescu and Jiang gave the first polynomial-time algorithms on a line and a closed cycle using linear programming~\cite{Dumitrescu2011a}.
Li and Wang improved these results to optimal $O(n)$-time algorithms on both a line and a closed cycle~\cite{Wang2025}.
Ghadiri and Yazdanbod gave an $O(n\log n)$-time algorithm for min-sum {\spreading} on a line~\cite{Ghadiri2016}.
The authors of~\cite{HonoratoDroguett2026} later gave an $O(n\log n)$-time algorithm on a closed cycle, completing the four cases of {\spreading} considered here.
A related min-max problem asks to move intervals on a line until they are pairwise disjoint.
Li and Wang gave an algorithm with optimal runtime of $O(n\log n)$ for this problem~\cite{Li2019}.
This problem is equivalent to min-max {\spreading} with individual thresholds when each point represents the centre of an interval.

For min-max {\connectivity} on a line, Li, Yan, and Zhang gave an algorithm that runs in linear time~\cite{Li2019A}.
The closed-cycle version can also be viewed as covering the closed cycle with sensors of equal sensing range.
The threshold is then twice the sensing range.
An $O(n)$-time algorithm for this setting was given in~\cite{Chen2013}.
Li later gave a linear-time algorithm that also obtains a lexicographically optimal solution~\cite{Li2026}.
To the best of our knowledge, no $O(n\log n)$-time algorithms had been given for the linear and cyclic versions of min-sum {\connectivity}.

Other work considers the related problem of moving sensors so that their sensing ranges cover a fixed line segment.
For the min-sum objective and equal sensing ranges, the authors of~\cite{Czyzowicz2010} gave an $O(n^2)$-time algorithm, later improved by Andrews and Wang to optimal $O(n\log n)$ time~\cite{Andrews2016}.
For the min-max objective, an $O(n^2)$-time algorithm was given in~\cite{Czyzowicz2009}.
This bound was improved to $O(n\log n)$ and to $O(n)$ when all sensors are initially on the barrier~\cite{Chen2013}.
For arbitrary sensing ranges, the min-sum problem is hard to approximate within any constant factor~\cite{Czyzowicz2010}.
The authors of~\cite{Gaspers2017} strengthened this hardness and gave approximation and parameterised results.
Mehrandish, Narayanan, and Opatrny instead considered minimising the number of sensors moved, giving polynomial-time algorithms for equal ranges and proving hardness for unequal ranges~\cite{Mehrandish2011}.

\subparagraph{Our Results}
We give transformations between {\connectivity} and {\spreading} in both directions.
These transformation apply to both objectives and the obtained optimal costs differ only by a scaling factor.
We first prove the result on a closed cycle (\Cref{sec:cyclic_case}) and then on a line (\Cref{sec:linear_case}).
Together with the known min-sum algorithms for {\spreading}, they imply that min-sum {\connectivity} is solvable in $O(n\log n)$ time on both domains.
We also extend the linear result to individual thresholds when the final positions are required to preserve their initial order (\Cref{sec:individual_thresholds}).
Finally, we give open problems in \Cref{sec:open_problems}.
\Cref{tab:summary} summarises the known runtimes together with our results.

\begin{table}[t]
    \centering
    \caption{Overview of runtimes for {\connectivity} and
            {\spreading}. Results marked with $\star$ are the runtimes implied by the results of this paper.}
    \label{tab:summary}
    \begin{tabular}{lccc}
        \toprule
        \bfseries Domain & \bfseries Objective
                         & \bfseries \connectivity                  & \bfseries \spreading \\
        \midrule
        \multirow{2}{*}{Line}
                         & min-max
                         & $O(n)$~\cite{Li2019A}
                         & $O(n)$~\cite{Wang2025}                                          \\
                         & min-sum
                         & $O(n\log n)$~$\star$
                         & $O(n\log n)$~\cite{Ghadiri2016}                                 \\
        \midrule
        \multirow{2}{*}{Closed cycle}
                         & min-max
                         & $O(n)$~\cite{Chen2013,Li2026}
                         & $O(n)$~\cite{Wang2025}                                          \\
                         & min-sum
                         & $O(n\log n)$~$\star$
                         & $O(n\log n)$~\cite{HonoratoDroguett2026}                        \\
        \bottomrule
    \end{tabular}
\end{table}

\section{Preliminaries}
Throughout the paper, we use $[n]$
as shorthand for the set $\{1,\ldots,n\}$ and we assume that $n\ge 2$.
For two points $p,q\in \mathbb{R}$, we define the \emph{distance} of $p$ and $q$ by $d(p,q) = |p-q|$.
For an $n$-tuple $X = (x_i)_{i\in[n]}$ of positions on the line with indices following order from left to right,
we also define $g_i(X) =  x_{i+1}-x_i$ for $i \in [n-1]$ and call these values the \emph{gaps} of $X$.
For the cyclic domain, let $C$ be a closed cycle of length $|C|$. We represent positions on $C$ by coordinates modulo $|C|$. For two points $p,q\in C$, the distance between $p$ and $q$ is $d(p,q) = \min\{c(p,q),c(q,p)\}$, where $c(p,q)$ is the \emph{length of the clockwise arc} from $p$ to $q$.
For an $n$-tuple $X = (x_i)_{i\in [n]}$ indexed in clockwise order with $x_1\le\cdots\le x_n$, the gaps of $X$ are $g_i(X) =x_{i+1}-x_i$ for $i\in [n]$, where $x_{n+1} = x_1 + |C|$ represents the same point as $x_1$ modulo $|C|$.
Note that the gaps satisfy $\sum_{i \in [n]} g_i(X) = |C|$.
Given a point $x\in C$, we \emph{unwrap} $C$ at $x$ by cutting it at $x$ and mapping each point $p\in C$ to the coordinate $c(x,p)\in[0,|C|)$.
We \emph{wrap} these coordinates back into $C$ by taking them modulo $|C|$.
Throughout the paper, we assume that the given tuples follow the corresponding index order, unless stated otherwise.

\subsection{Problem Formulation}

We give the general formulation that includes both {\spreading} and {\connectivity} problems.
Let $P = (p_i)_{i\in [n]}$ and $Q = (q_i)_{i\in [n]}$ denote the initial and final positions, respectively, of $n$ labelled points on a line or closed cycle.
For each $i\in[n]$, point $i$ moves from $p_i$ to $q_i$.
We consider the movement cost functions:
\begin{equation*}
    \costone(P,Q)
    =\sum_{i \in [n]}d(p_i,q_i),
    \quad
    \costinf(P,Q)
    =\max_{i\in[n]}d(p_i,q_i).
\end{equation*}

Given a \emph{threshold} $r >0$, we denote an instance on the line or closed cycle by $(P,r)$.
The {\connectivity} problem asks to find $Q$ such that every gap $g_i(Q)$ is {\em at most} $r$.
The {\spreading} problem, instead, asks to find $Q$ such that every gap $g_i(Q)$ is {\em at least} $r$.
In either problem, the goal is to find $Q$ while minimising $\costone(P,Q)$ (min-sum) or $\costinf(P,Q)$ (min-max).
When the domain, the threshold $r$, and the objective are clear from context, we denote the minimum cost among all feasible $Q$ by $\OPTconn(P)$ for {\connectivity} and by $\OPTspread(P)$ for {\spreading}.
For the cyclic version, we also assume that $|C|\le nr$ for {\connectivity} and $|C|\ge nr$ for {\spreading}.
There are no feasible solutions for the corresponding instance otherwise.
Lastly, we assume for the cyclic version that $p_1 = 0$ without loss of generality since we can choose $p_1$ as the origin and all gaps remain unchanged.
\subparagraph{Order Preservation} A solution $Q=(q_i)_{i\in[n]}$ preserves the order of $P$ if $q_1\le\cdots\le q_n$ on the line and if $q_1,\ldots,q_n$ occur in the same cyclic order as $p_1,\ldots,p_n$ on a closed cycle.
For both problems defined above, we assume that there is an optimal solution preserving the order of the points on both a line and a closed cycle.
This was shown by Dumitrescu and Jiang~\cite{Dumitrescu2011a}, and an analogous property was shown by Ghadiri and Yazdanbod~\cite{Ghadiri2016} for the linear version of min-sum {\spreading}.
In general, the same uncrossing argument applies to all cases considered here, since one can always swap two crossing points while preserving feasibility and not increasing the movement cost.
On a closed cycle, we first unwrap the portion containing the crossing and then wrap it back.

\section{Transformation for the Cyclic Version}
\label{sec:cyclic_case}
We show how to transform an instance of {\connectivity} to an instance of {\spreading}, and vice versa.
That is, given an instance $(P,r)$ with optimal solution $Q$ for {\connectivity}, we show that there is a corresponding instance $(P',r)$ with optimal solution $Q'$.
\Cref{fig:reduction_example} illustrates the transformation and the corresponding solutions for an instance with four points.
We show how to derive these values.
For simplicity, we may assume that $r=1$.
If $r \neq 1$, we can divide all coordinates and the length of the closed cycle by $r$ to produce a scaled instance.
Multiplying an optimal solution of the resulting instance by $r$ recovers an optimal solution to the original instance.

The rest of the section is devoted to proving the following theorem.
\begin{theorem}\label{thm:conn-edg-cyclic}
    For $r>0$, one can transform an instance $(P,r)$ of {\connectivity} on a closed cycle $C$ into an instance $(P',r)$ of {\spreading} such that $\OPTconn(P)=(|C|/r)\OPTspread(P')$, and vice versa.
\end{theorem}
\begin{proof}
    By the scaling described above, it suffices to prove the statement for $r=1$.
    Our transformation is based on the following observation: by the order preservation property, the constraints of both problems depend exclusively on the gaps between consecutive points.
    Thus if we want to transform an instance $(P,1)$ of one problem, say {\connectivity}, we must only establish a correspondence between the gaps of $P$ and the constructed instance $(P',1)$ of {\spreading}.
    If such a correspondence is possible, then a movement that reduces a gap in one instance must correspond to a movement that increases the corresponding gap in the other instance, and vice versa.
    The factor relating these distances consequently scales both the total and the maximum movement.

    We show that the above factor exists.
    In particular, we want an instance $(P',1)$ such that the $i$th gap $g_i(P)$ satisfies $g_i(P) \le 1$ if and only if the $i$th gap $g_i(P')$ satisfies $g_i(P') \ge 1$ for all $i\in [n]$, or equivalently, $g_i(P) - 1 \le 0$ if and only if $g_i(P') - 1 \ge 0$.
    For the constructed instance, we first consider setting $g_i(P) -1 = -g_i(P')+1$ since this assignment satisfies the inequality.
    However, we may have that $g_i(P) >2$ which yields a negative value for $g_i(P')$.
    To solve this, we scale the value of $g_i(P)-1$ by $1/|C|$, yielding a value bounded by $1$ since no gap can be larger than $|C|$.
    We set then $(g_i(P) - 1)/|C| = -(g_i(P')-1)$ which gives $g_i(P') = (|C|+1-g_i(P))/|C|$.

    \begin{figure}[bt]
        \centering
        \includegraphics[scale=1,page=2]{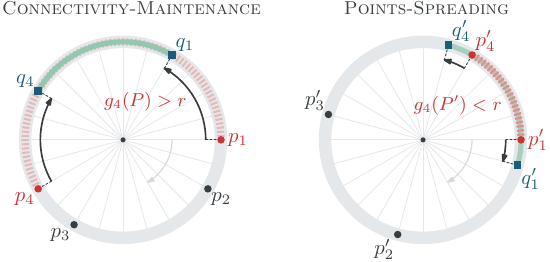}
        \caption{Transformation example. Here, $P = (p_1,\ldots,p_4) = (0,2,8,10)$, $r = 6$ and $|C| = 24$. An optimal solution for both objectives is $Q = (q_1,\ldots,q_4)=(20,2,8,14)$ (only $q_1$ and $q_4$ are shown). The {\spreading} instance is $(P',r)$, where $P'=(0,7,13,20)$ and $|C'| =24$, and an optimal solution is $Q' = (1,7,13,19)$.
            The $4$th gap violates the respective feasibility condition in both instances. In particular, $g_4(P) = 14 > 6$ and $g_4(P') = 4 < 6$.
            Their movement vectors are $(-4,0,0,4)$ and $(1,0,0,-1)$, and the second is obtained by multiplying the first by $-r/|C|$.
        }
        \label{fig:reduction_example}
    \end{figure}

    We derive the remaining positions using the obtained gap formula for $P'$.
    For $i\in [n-1]$, we have $g_i(P')  = p'_{i+1} - p'_i = (|C|+1-p_{i+1}+p_i)/|C|$ and obtain the recurrence $p_{i+1} + |C| p'_{i+1} = p_i + |C|p'_i + |C| + 1$.
    Expanding this recurrence yields $p_{i} + |C| p'_{i} = p_1 + |C| p'_1 + (i-1)(|C|+1)$ for every $i\in [n]$.
    Now, we choose $p'_1 =0$ and obtain $p'_{i} = ((i-1)(|C|+1)-p_i)/|C|$ for $i \in [n]$ since $p_1 = 0$.
    We can choose $p'_1 = 0$ since $P'$ is also determined only by the gaps between points.

    For $C'$, we have $|C'| = \sum_{i \in [n]} g_i(P') = \sum_{i \in [n]} (|C|+1-g_i(P))/|C|$. Expanding this formula, we obtain:
    \begin{align*}
        \sum_{i \in [n]} & (|C|+1-g_i(P))/|C| = \frac{1}{|C|}\Bigl(\sum_{i \in [n]} (|C|+1) - \sum_{i \in [n]} g_i(P)\Bigr) \\
                         & = \frac{1}{|C|}(n(|C|+1) - ((p_2-p_1)+\cdots+(p_{n}-p_{n-1})+(|C|+p_1-p_n)))                     \\ &= \frac{1}{|C|}(n(|C|+1)- |C|).
    \end{align*}
    Hence $|C'| = (1/|C|)(n(|C|+1)- |C|)$.
    Since $|C|\le n$, the value of $|C'|$ is at least $n$.
    Thus $P'$ can be spread on $C'$.

    Lastly, we construct $Q'$ using the given solution $Q$ for $P$.
    For simplicity, we call a chosen shortest arc from $p_i$ to $q_i$ the \emph{movement arc} of point $i$ and say that a movement arc \emph{crosses} a point if the point is in the interior of the arc.
    To accomplish this, we need a solution $Q$ satisfying $d(p_i,q_i)=|q_i-p_i|$ for every $i\in [n]$.
    This equality may fail when the movement arc of point $i$ crosses the origin, in which case $d(p_i,q_i)$ equals $|C|-|q_i-p_i|$.
    We first show that in an optimal solution, there is always a point not crossed by any movement arc.
    \begin{claim}\label{claim:not_crossing_origin}
        There is an optimal solution $Q$ and a point $x \in C$ such that, for every $i\in[n]$, the movement arc of point $i$ does not cross $x$.
    \end{claim}
    \begin{claimproof}
        Let $Q$ be an optimal solution with minimum total movement.
        By the order preservation property, such an optimal solution always exists.

        First, we show that no point is crossed by both a clockwise and a counterclockwise movement arc.
        Suppose the contrary: there is a point $x$ crossed by both a clockwise and a counterclockwise movement arc.
        Let point $i$ move clockwise from $p_i$ to $q_i$ and point $j$ move counterclockwise from $p_j$ to $q_j$ (see~\Cref{fig:non_crossing_movemet} (left)).
        We unwrap $C$ at $x+|C|/2$ and assume going from left to right is going clockwise in $C$.
        Since the movement arcs have length at most $|C|/2$, both movement arcs crossing $x$ lie entirely in the unwrapped interval.
        If a movement distance equals $|C|/2$, we choose as the movement arc the shortest arc that crosses $x$.
        Now, let $a$ be the distance from $p_i$ to $x$, and $b$ the distance from $x$ to $q_i$.
        Similarly, let $c$ be the distance from $p_j$ to $x$, and $e$ the distance from $x$ to $q_j$.
        We choose $x$ as the origin and observe that the points move from $-a$ to $b$ and from $c$ to $-e$, respectively.
        \Cref{fig:non_crossing_movemet} (top right) illustrates this.
        If we swap these values (that is, we swap $q_i$ and $q_j$), the movement distances become $|a-e|$ and $|c-b|$ (\Cref{fig:non_crossing_movemet} (bottom right)).
        Moreover, swapping $q_i$ and $q_j$ leaves the set of final positions unchanged and hence preserves feasibility.
        On the other hand, $\max\{|a-e|,|c-b|\} $ is at most $ \max\{a+b,c+e\}$, while $|a-e|+|c-b|$ is less than $(a+b)+(c+e)$, contradicting the optimality of $Q$ for the min-sum objective and the choice of $Q$ for the min-max objective.

        \begin{figure}[bt]
            \centering
            \includegraphics[scale=1,page=5]{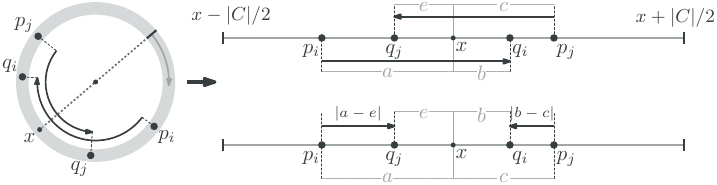}
            \caption{No point is crossed by both a clockwise and a counterclockwise movement arc. (Left) A point $x$ is crossed by both the clockwise movement arc of point $i$ from $p_i$ to $q_i$ and the counterclockwise movement arc of point $j$ from $p_j$ to $q_j$.
                (Right) Top: Unwrapping $C$ at $x+|C|/2$ gives movements from $-a$ to $b$ and from $c$ to $-e$.
                Bottom: Swapping $q_i$ and $q_j$ gives movement distances $|a-e|$ and $|c-b|$.
            }
            \label{fig:non_crossing_movemet}
        \end{figure}

        We now consider whether every point of $C$ is crossed by a movement arc.
        If not, there is at least one point not crossed by any movement arc and we are done.
        Then, suppose that every point of $C$ is crossed by a movement arc.
        Since no point is crossed by both a clockwise and a counterclockwise movement arc, either every point moves clockwise or every point moves counterclockwise.
        Without loss of generality, we thus may assume that every point moves clockwise.

        For every $i\in [n]$, we note that $q_{i-1}$ ($q_n$ when $i=1$) must be contained in the movement arc of point $i$; see \Cref{fig:not_crossing_origin}~(left)
        Otherwise, some point on the arc from $q_{i-1}$ to $p_i$ is not crossed by any movement arc.
        We then consider the solution $S = (q_n,q_1,\ldots,q_{n-1})$ where point $i$ is moved to $q_{i-1}$ instead (\Cref{fig:not_crossing_origin}~(middle)).
        This reassignment preserves the order of $P$ and does not change the components of $Q$, thus $S$ is also a feasible solution.
        Moreover, the distances satisfy $d(p_i,q_{i-1})\le d(p_i,q_i)$ for every $i\in [n]$ and strict inequality for at least one $i$.
        Hence $\costone(P,S) < \costone(P,Q)$ holds, contradicting the optimality of $Q$ for the min-sum objective.
        For the min-max objective, the reassignment implies $\costinf(P,S) \le \costinf(P,Q)$.
        If the inequality is strict, then the optimality of $Q$ is contradicted.
        Otherwise, $S$ is also optimal.
        If some point of $C$ is not crossed by any movement arc of $S$, then we obtained a desired optimal solution (\Cref{fig:not_crossing_origin}~(right)).
        Otherwise, we repeat the reassignment.
        Since each reassignment strictly decreases the total movement and there are only $n-1$ reassignments (if we do the $n$th reassignment we obtain $Q$ with less total movement, a contradiction), we eventually obtain an $S$ where a point is not crossed by any movement arc.
        Therefore there is a point $x\in C$ not crossed by any movement arc.
    \end{claimproof}

    \begin{figure}[bt]
        \centering
        \includegraphics[scale=1,page=1]{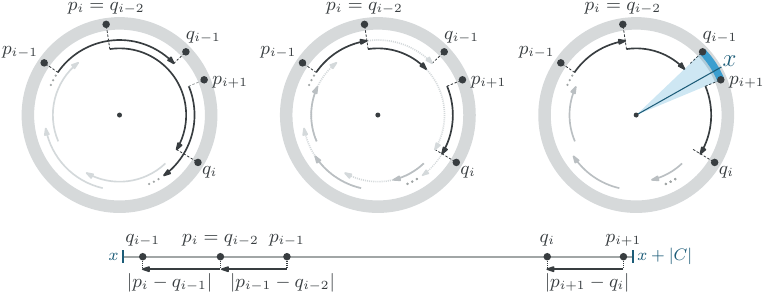}
        \caption{Illustration of \Cref{claim:not_crossing_origin}. (Top)
            Left: Every point of $C$ is crossed by a clockwise movement arc and $q_{i-1}$ lies on the movement arc of point $i$.
            Middle: Point $i$ is moved to $q_{i-1}$.
            Right: After reassignment, there is a point $x$ not crossed by any movement arc.
            (Bottom) The closed cycle $C$ is unwrapped at $x$, showing that all distances have the form $|p-q|$.
        }
        \label{fig:not_crossing_origin}
    \end{figure}

    By \Cref{claim:not_crossing_origin}, after unwrapping $C$ at $x$ (and, if necessary, shifting the indices), we may choose coordinates such that $p_{n+1}=p_1+|C|$, $q_{n+1}=q_1+|C|$, and $d(p_i,q_i)=|q_i-p_i|$ for every $i\in[n]$; see \Cref{fig:not_crossing_origin} (bottom).
    Translating all positions, we may still assume that $p_1 = 0$.
    Moreover, for every $i\in [n]$, we write $q_i = p_i + d_i$ for a value $d_i$.
    Similarly, we write $q'_i = p'_i + d'_i$.
    Thus, the point $p'_i$ equals $q'_i - d'_i = ((i-1)(|C|+1) - p_i)/|C| $, which in turn equals $((i-1)(|C|+1) - q_i + d_i)/|C|$.
    Using this formula, we conveniently set $q'_i$ to $((i-1)(|C|+1) - q_i)/|C|$ and $d'_i$ to $- d_i/|C|$.
    The obtained values of $Q'$ are feasible for the following reason.
    For every $i\in [n]$, the gap $g_i(Q')$ equals $q'_{i+1} - q'_i = (|C|+1 -q_{i+1}+q_i)/|C|$.
    Since $Q$ is feasible, we have $g_i(Q) = q_{i+1}-q_i \le 1$ for every $i \in [n]$.
    Hence $-q_{i+1}+q_{i} \ge -1$ holds, implying that $g_i(Q')\ge 1$.

    We showed how to construct an instance $(P',1)$ of {\spreading} and a feasible solution $Q'$ given an instance $(P,1)$ of {\connectivity} with solution $Q$.
    Note that the same construction applies in the opposite direction.
    If $(P,1)$ is an instance of {\spreading}, then $g_i(Q)\ge 1$ if and only if $g_i(Q')\le 1$. Moreover, $|C|\ge n$ implies $|C'|\le n$.
    Thus, the constructed instance $(P',1)$ is an instance of {\connectivity}.

    We only need to show that $Q$ is optimal if and only if the corresponding $Q'$ is optimal.
    Recall that the distance $d(p'_i, q'_i) $ equals $\min\{|d'_i|, |C'| - |d'_i|\} \le |C'|/2$.
    Since $|d_i|\le |C|/2$, we have that $|d'_i| \le 1/2$.
    Moreover, we have $|C'|=n-1+n/|C|\ge 1$ implied by $n\ge 2$.
    Hence $|d'_i|\le 1/2\le |C'|/2$, so the distance $d(p'_i,q'_i)$ equals $|d'_i|$.
    This gives us the following cost for the min-sum objective,
    \begin{equation*}
        \costone(P,Q) = \sum_{i\in [n]}|d_i| = \sum_{i\in [n]}|(-|C|d'_i)| = |C|\sum_{i\in [n]}|d'_i| = |C|\costone(P',Q').
    \end{equation*}
    For the min-max objective,
    \begin{equation*}
        \costinf(P,Q) = \max_{i\in [n]} |d_i| = \max_{i\in [n]} |(-|C|d'_i)| = |C|\max_{i\in [n]} |d'_i| = |C|\costinf(P',Q').
    \end{equation*}
    The above equalities show that the costs of the corresponding solutions (for either objective) differ by a factor of $|C|$.
    Assume that $Q$ is optimal.
    If $Q'$ is not optimal, let $Q'' = (q''_i)_{i\in [n]}$ be an optimal solution for $P'$.
    Using \Cref{claim:not_crossing_origin} and after unwrapping $C'$, we may assume $q''_i=p'_i+d''_i$ with $d(p'_i,q''_i)=|d''_i|$ for every $i\in[n]$, where $p'_{n+1}$ equals $p'_1+|C'|$, $q''_{n+1}$ equals $q''_1+|C'|$, and $d''_{n+1}$ equals $d''_1$.
    Using the equality $d'_i=-d_i/|C|$ for the opposite transformation, we consider moving point $i$ from $p_i$ by $-|C|d''_i$.
    For every $i\in [n]$, the difference between the corresponding consecutive positions satisfies $(p_{i+1}-|C|d''_{i+1})-(p_i-|C|d''_i) = g_i(P)-|C|(d''_{i+1}-d''_i)$ which in turn equals $|C|+1-|C|g_i(Q'')$ since $d''_{i+1}-d''_i = g_i(Q'')-g_i(P')$.
    If $P$ is an instance of {\connectivity}, then $Q''$ is a solution for {\spreading}, so $g_i(Q'')\ge 1$ and $|C|+1-|C|g_i(Q'')$ is at most $1$.
    Also, since $p_{n+1}=p_1+|C|$ and $d''_{n+1}=d''_1$, the total of these differences is $|C|$ and hence any empty gap greater than $1$ found by sweeping from $p_1-|C|d''_1$ to $p_1-|C|d''_1+|C|$ would contradict this bound.
    Analogously, if $P$ is a instance of {\spreading}, then $Q''$ is a solution for {\connectivity}, so $g_i(Q'')\le 1$ and $|C|+1-|C|g_i(Q'')$ is at least $1$.
    Hence the constructed solution is feasible for $P$ in either direction.
    Moreover, the distance for every $i\in [n]$ satisfies $d(p_i,p_i-|C|d''_i) \le |C|\,|d''_i| = |C|\,d(p'_i,q''_i)$.
    In other words, the obtained solution has cost at most $|C|$ times the cost of $Q''$, for either objective.
    Since the cost of $Q''$ is smaller than the cost of $Q'$, this gives a solution for $P$ with smaller cost than $Q$, contradicting optimality.
    In the other direction, if $Q'$ is optimal but $Q$ is not, applying the construction of $Q'$ to an optimal solution for $P$ gives a solution for $P'$ with cost smaller than $Q'$, again contradicting optimality.
    Therefore for either objective, $Q$ is optimal if and only if the corresponding $Q'$ is optimal.
\end{proof}

The above theorem shows that the optimal costs of the corresponding {\connectivity} and {\spreading} instances differ by a factor of $|C|/r$, for both objectives. The authors of~\cite{HonoratoDroguett2026} proved that the cyclic version of min-sum {\spreading} on $n$ points can be solved in $O(n\log n)$ time.
Their algorithm removes overlaps of open unit circular arcs while minimising the total movement, and the reduction centres an arc at each point.
Two arcs are disjoint if and only if their centres are at distance at least $1$, and moving an arc is done by moving its centre.
Our transformation takes linear time, so \Cref{thm:conn-edg-cyclic} implies the following:

\begin{corollary}
    The cyclic version of min-sum {\connectivity} is solvable in $O(n\log n)$ time.
\end{corollary}

\section{Transformation for the Linear Version}
\label{sec:linear_case}

We now show how to transform an instance of {\connectivity} on the line into an instance of {\spreading}, and vice versa.
We use the same transformation idea as in \Cref{thm:conn-edg-cyclic}.
Unlike the cyclic version, the gaps between consecutive points can have arbitrary length in the linear version.
Thus, applying the transformation to a solution of {\spreading} with a sufficiently large gap may reverse the order of the corresponding points.
However, we show that the gaps of an optimal solution are bounded by the largest initial gap.
In particular, every gap is at most the maximum of $r$ and the largest gap of the initial instance.
This lets us apply an analogous transformation for the linear version.
Let $(P,r)$ be an instance on the line, where $P=(p_i)_{i\in[n]}$, and let $M=1+\max\{1,\max_{i\in[n-1]}g_i(P)/r\}$.
We first show that an optimal solution of {\spreading} can be chosen with gaps between $r$ and $(M-1)r$.

\begin{figure}[bt]
    \centering
    \includegraphics[scale=1,page=3]{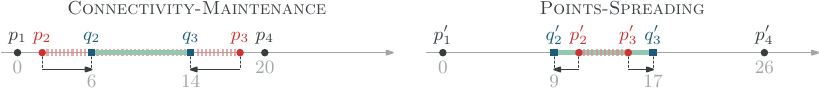}
    \caption{Transformation example on the line with $r=8$ and $M=3$.
        The {\connectivity} instance is $(P,r)$, where $P=(0,2,18,20)$, and has an optimal solution $Q=(0,6,14,20)$ for both objectives.
        The corresponding {\spreading} instance is $(P',r)$, where $P'=(0,11,15,26)$, and its optimal solution is $Q'=(0,9,17,26)$.
        The second gap violates feasibility in both initial instances: $g_2(P)=16>8$ and $g_2(P')=4<8$.
        The movement vectors $(0,4,-4,0)$ and $(0,-2,2,0)$ show that movements reverse direction, while $d(p_i',q_i')=d(p_i,q_i)/(M-1)$.
    }
    \label{fig:reduction_example_line}
\end{figure}

\begin{lemma}\label{lem:bounded_gap_line}
    For either objective, there is an order-preserving optimal solution $Q=(q_i)_{i\in[n]}$ for the {\spreading} instance $(P,r)$ such that $r\le g_i(Q)\le (M-1)r$ for every $i\in[n-1]$.
\end{lemma}
\begin{proof}
    Again, it suffices to prove the statement for $r=1$.
    Let $Q$ be an optimal solution preserving the order of $P$.
    Suppose that $g_i(Q)>M-1$ for some $i\in[n-1]$, and let $\lambda=g_i(Q)-(M-1)$.
    We divide the proof into three cases.

    First, suppose that $p_i<q_i$ (\Cref{fig:bounded_gap_line}, left).
    Since $g_i(P)\le M-1$, we have $(q_{i+1}-p_{i+1})-\lambda \ge q_i-p_i>0$.
    We obtain a solution by moving point $i+1$ to the left by $\lambda$.
    In this solution, the $i$th gap becomes $M-1$ and other gaps remain feasible.
    Moreover, $q_{i+1}$ remains to the right of $p_{i+1}$, so the distance $q_{i+1}-p_{i+1}$ strictly decreases.
    The case $q_{i+1}<p_{i+1}$ (\Cref{fig:bounded_gap_line}, middle) is analogous by moving point $i$ to the right by $\lambda$.

    It remains to consider the case $q_i\le p_i$ and $p_{i+1}\le q_{i+1}$, see \Cref{fig:bounded_gap_line} (right).
    Since $p_{i+1}-p_i\le M-1$ and $q_{i+1}-q_i> M-1$, there is an $x$ such that $[p_i,p_{i+1}]\subseteq[x,x+M -1]\subseteq[q_i,q_{i+1}]$.
    Thus, we consider moving points $i$ and $i+1$ to $x$ and $x+M-1$, respectively.
    Again, the $i$th gap becomes $M-1$ and other gaps remain feasible.
    Moreover, we move points $i$ and $i+1$ closer to their initial positions, hence at least one distance strictly decreases.

    In all three cases, the obtained solution is feasible and has smaller total movement.
    For the min-sum objective, this contradicts the optimality of $Q$.
    For the min-max objective, we assume that $Q$, among all order-preserving optimal solutions, is a solution with minimum total movement.
    In the above construction, no movement increases, so the obtained solution is also optimal for the min-max objective but contradicts the choice of $Q$ since we obtained smaller total movement.
    Therefore, $g_i(Q)\le M-1$ for every $i\in[n-1]$.
\end{proof}

\begin{figure}[bt]
    \centering
    \includegraphics[scale=1,page=4]{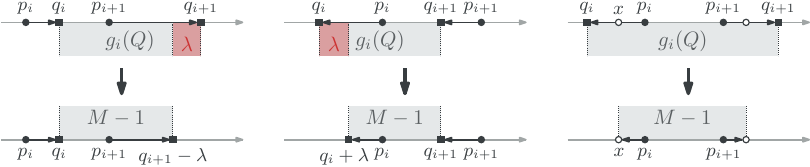}
    \caption{ The three cases of~\Cref{lem:bounded_gap_line} for $r=1$. In all cases, one can obtain a gap bounded by $M-1$ and smaller movement by shifting a point towards its initial position.
    }
    \label{fig:bounded_gap_line}
\end{figure}

In the cyclic version, the length $|C|$ of the closed cycle gives an upper bound for the gaps and we exploit this fact to scale the values in the transformation.
On the line, no such bound exists in general.
However, with \Cref{lem:bounded_gap_line} we showed that there is an optimal solution where each gap is bounded by $(M-1)r$.
After scaling by $1/r$, we can use $M-1$ for the linear version in the same manner as $|C|/r$ in the cyclic version.
This gives the following result.

\begin{theorem}
    \label{thm:conn-edg-line}
    For $r>0$, one can transform an instance $(P,r)$ of the linear version of {\connectivity} into an instance $(P',r)$ of {\spreading} such that\ $\OPTconn(P)=(M-1)\OPTspread(P')$, and vice versa.
\end{theorem}
\begin{proof}
    We prove the statement for $r= 1$.
    Let $(P,1)$ be an instance of {\connectivity} and let $Q=(q_i)_{i\in[n]}$ be an optimal solution preserving the order of $P$.
    We use the same construction as the one for \Cref{thm:conn-edg-cyclic} replacing $|C|+1$ by $M$.
    For every $i \in [n]$, the positions $p'_i$ in the constructed instance $(P',1)$ and $q'_i$ in its solution $Q' = (q'_i)_{i\in [n]}$ for {\spreading} are defined by:
    \begin{equation*}
        p_i'=\frac{(i-1)M-p_i}{M-1}, \qquad q_i'=\frac{(i-1)M-q_i}{M-1}.
    \end{equation*}
    Similarly, the gaps are defined as:
    \begin{equation*}
        g_i(P')=\frac{M-g_i(P)}{M-1},\qquad g_i(Q')=\frac{M-g_i(Q)}{M-1},
    \end{equation*}
    for every $i \in [n-1]$.
    We note that $Q'$ is feasible for the following reason.
    Since $Q$ is feasible, we have that $g_i(Q)\le 1$ and $g_i(Q')\ge 1$.
    Thus $Q'$ is a feasible solution of {\spreading}.
    Moreover, $d(p'_i,q'_i) = d(p_i,q_i)/(M-1)$ for every $i\in [n]$.
    Hence, for either objective, the cost of $Q'$ is $1/(M-1)$ times the cost of $Q$, implying that $\OPTspread(P') \le 1/(M-1)\OPTconn(P)$.

    In the other direction, let $Q'' = (q''_i)_{i\in [n]}$ be an optimal solution of {\spreading} for $P'$. Since the gap $g_i(P')$ satisfies $g_i(P') \le M/(M-1)$, we may assume that $1 \le g_i(Q'') \le M/(M-1)$ by \Cref{lem:bounded_gap_line}.
    Using the inverse of the above formula, we obtain $q_i = (i-1)M - (M-1)q''_i$ and the resulting gaps are $g_i(Q) = M-(M-1)g_i(Q'')$.
    By the bounds given for each $g_i(Q'')$, the resulting gaps are between $0$ and $1$.
    Thus $Q$ is a feasible solution of {\connectivity}.
    Moreover, $d(p_i,q_i)=(M-1)d(p_i',q''_i)$ for every $i \in [n]$, so the cost of $Q$ equals $(M-1)$ times the cost of $Q''$.
    Consequently, we obtain $\OPTconn(P) \le (M-1)\OPTspread(P')$.
    Combining both inequalities, we obtain $\OPTconn(P)=(M-1)\OPTspread(P')$.

    The transformation from an instance $(P,1)$ of {\spreading} to an instance $(P',1)$ of {\connectivity} follows analogously.
    In particular, by \Cref{lem:bounded_gap_line}, we may choose an optimal solution $Q$ of {\spreading} whose gaps are between $1$ and $ M-1$.
    Following the above transformation we conclude that the constructed $Q'$ is feasible for {\connectivity}.
    In the other direction, every feasible solution $Q''$ for $P'$ produces a feasible solution for $P$.
    For every $i\in[n]$, the movement of the corresponding point in the resulting solution equals $(M-1)d(p_i',q''_i)$.
    Hence we obtain $\OPTspread(P)=(M-1)\OPTconn(P')$.
    We proved the statement for $r = 1$.
    For arbitrary $r>0$, the result follows by scaling the coordinates by $1/r$ before applying the construction and scaling them back by $r$.
    This concludes the proof.
\end{proof}

As we mentioned in the introduction, the $O(n\log n)$-time algorithm of Ghadiri and Yazdanbod~\cite{Ghadiri2016} solves the linear version of min-sum {\spreading}. Together with \Cref{thm:conn-edg-line}, we conclude the following.

\begin{corollary}
    The linear version of min-sum {\connectivity} is solvable in $O(n\log n)$ time.
\end{corollary}

\section{Individual Thresholds on a Line}
\label{sec:individual_thresholds}

In this section, we extend the equivalence between {\connectivity} and {\spreading} on the line to the setting where points have {\em individual thresholds}.
That is, instead of using the same threshold $r$ for every pair of points, point $i$ has an individual threshold $r_i>0$.
For points $i$ and $j$, we define their threshold as $r_{i,j}=(r_i+r_j)/2$ and their gap in a solution $Q$ as $g_{i,j}(Q)=d(q_i,q_j)$.
If $r_i=r$ for every $i\in[n]$, then $r_{i,j}=r$, the adapted problems reduce to the uniform versions.
Throughout this section, we refer to the usual problems with a common threshold as the \emph{uniform versions}.

Unlike in the uniform versions, we cannot assume that an optimal solution preserves the order of the initial points.
For min-max {\spreading} with individual thresholds, Li and Wang gave an $O(n\log n)$-time algorithm for this problem in the context of intervals~\cite{Li2019}.
Here, we focus on a restricted version in which the points must preserve their initial order and establish the equivalence with {\connectivity} for both objectives.
This allows us to restrict the constraints to gaps between consecutive points in this order.
\subparagraph{Problem Formulation} Let $P=(p_i)_{i\in[n]}$ denote the initial positions of $n$ {\em labelled} points on the line, indexed so that $p_1\le\cdots\le p_n$.
Given $P$ and $(r_i)_{i\in[n]}$, the adapted version of {\connectivity} asks to find an $n$-tuple $Q=(q_i)_{i\in[n]}$ such that $q_1\le\cdots\le q_n$ and $g_{i,i+1}(Q)\le r_{i,i+1}$ for every $i\in[n-1]$.
The adapted version of {\spreading}, instead, asks to find such a tuple $Q$ satisfying $g_{i,i+1}(Q)\ge r_{i,i+1}$ for every $i\in[n-1]$.
As before, the goal of both problems is to find $Q$ while minimising $\costone(P,Q)$ or $\costinf(P,Q)$.

Again, the transformation uses the same idea as in the uniform linear version.
We apply the construction to $P$ and $Q$, obtaining $P'$ and $Q'$, respectively.
For each $i\in[n-1]$, the construction ensures that $g_{i,i+1}(P)\le r_{i,i+1}$ if and only if $g_{i,i+1}(P')\ge r_{i,i+1}$.
The same correspondence holds between the gaps of $Q$ and $Q'$.
Thus, every feasible solution of {\connectivity} gives a feasible solution of {\spreading}.
As in \Cref{lem:bounded_gap_line}, to establish the equivalence, we also need a common factor that bounds the gaps of an optimal solution with respect to their individual thresholds.

Let $M=1+\max\{1,\max_{i\in[n-1]}g_{i,i+1}(P)/r_{i,i+1}\}$.
The definition of $M$ implies that $g_{i,i+1}(P)\le (M-1)r_{i,i+1}$ for every $i\in[n-1]$.
We can adapt the proof of \Cref{lem:bounded_gap_line} and show that an optimal solution of {\spreading} satisfies this bound.
For every $i\in [n-1]$, whenever an optimal solution $Q$ has $g_{i,i+1}(Q)>(M-1)r_{i,i+1}$, we reduce this gap to $(M-1)r_{i,i+1}$ using the same three cases.
Since $g_{i,i+1}(P)\le (M-1)r_{i,i+1}$, the three cases move point $i$, point $i+1$, or both closer to their respective initial positions while reducing the gap. Hence, no movement increases and at least one strictly decreases.
Consequently, for either objective, there is an optimal solution $Q$ satisfying:
\begin{equation}
    r_{i,i+1}\le g_{i,i+1}(Q)\le (M-1)r_{i,i+1}, \qquad i\in[n-1].
    \label{eq:bounded_gap_line-individual}
\end{equation}

We use this value as in~\Cref{thm:conn-edg-line} and show the following result.

\begin{theorem}
    \label{thm:conn-edg-line-individual}
    One can transform an instance of the adapted linear version of {\connectivity} with initial positions $P$ into an instance of {\spreading} with initial positions $P'$ and the same individual thresholds such that $\OPTconn(P)=(M-1)\OPTspread(P')$, and vice versa.
\end{theorem}
\begin{proof}
    We use a very similar argument to that for~\Cref{thm:conn-edg-line}.
    Let the {\connectivity} instance have initial positions $P$.
    We first derive a recurrence for the positions in $P'$. 
    To this end, we establish a correspondence between the gaps of $P$ and $P'$.
    In particular, for every $i\in[n-1]$, we want
    \begin{equation*}
        g_{i,i+1}(P)-r_{i,i+1}\le 0
        \quad\text{if and only if}\quad
        g_{i,i+1}(P')-r_{i,i+1}\ge 0.
    \end{equation*}
    Following the same argument as in the proof of \Cref{thm:conn-edg-cyclic} for the scaling and using equation~\eqref{eq:bounded_gap_line-individual}, we scale $g_{i,i+1}(P)-r_{i,i+1}$ by $1/(M-1)$ and set
    $(g_{i,i+1}(P)-r_{i,i+1})/(M-1)=-(g_{i,i+1}(P')-r_{i,i+1})$, which gives $g_{i,i+1}(P')=(M r_{i,i+1}-g_{i,i+1}(P))/(M-1)$.

    We derive the positions using the obtained gap formula for $P'$.
    For every $i\in[n-1]$, we have that the $i$th gap of $P'$ equals
    $g_{i,i+1}(P')=p'_{i+1}-p'_i=(M r_{i,i+1}-p_{i+1}+p_i)/(M-1)$,
    and obtain the recurrence $p_{i+1}+(M-1)p'_{i+1}=p_i+(M-1)p'_i+M r_{i,i+1}$.
    Expanding this recurrence, we obtain
    \begin{align*}
        p_i+(M-1)p'_i
         & =p_{i-1}+(M-1)p'_{i-1}+M r_{i-1,i}                 \\
         & =p_1+(M-1)p'_1+M(r_{1,2}+r_{2,3}+\cdots+r_{i-1,i}) \\
         & =p_1+(M-1)p'_1+M\sum_{j=1}^{i-1}r_{j,j+1}.
    \end{align*}
    So we choose $p'_1=-p_1/(M-1)$ and obtain $p'_i=(M\sum_{j=1}^{i-1}r_{j,j+1}-p_i)/(M-1)$ for every $i\in[n]$.
    By the definition of $M$, these positions satisfy $g_{i,i+1}(P')\ge r_{i,i+1}/(M-1)>0$, so $P'$ follows the same index order.
    We abbreviate the sum by setting $R_i=\sum_{j=1}^{i-1}r_{j,j+1}$.
    Similarly, given an optimal solution $Q=(q_i)_{i\in[n]}$ for $P$, we apply the same construction and obtain $q'_i=(M R_i-q_i)/(M-1)$ for every $i\in[n]$.
    We note that $Q'$ is feasible by the following reason.
    Since $Q$ is feasible, we have that $g_{i,i+1}(Q)\le r_{i,i+1}$ and $g_{i,i+1}(Q')=(M r_{i,i+1}-g_{i,i+1}(Q))/(M-1)\ge r_{i,i+1}$ for every $i\in[n-1]$.
    Thus $Q'$ is a feasible solution of {\spreading}.

    For optimality, we use the same argument as in the uniform linear version.
    From the formulas for $p'_i$ and $q'_i$, we have
    $q'_i-p'_i=-(q_i-p_i)/(M-1)$, and thus the distance
    $d(p'_i,q'_i) $ equals $d(p_i,q_i)/(M-1)$ for every $i\in[n]$.
    Hence for either objective,
    we have $\OPTspread(P')\le \OPTconn(P)/(M-1)$.
    In the other direction, let $Q''=(q''_i)_{i\in[n]}$ be an optimal solution of {\spreading} for $P'$.
    We show how to transform $Q''$ into a feasible solution $Q$ for $P$.
    The gap formula gives $g_{i,i+1}(P')\le M r_{i,i+1}/(M-1)$ for every $i\in[n-1]$, and since $M/(M-1)\ge 1$, equation~\eqref{eq:bounded_gap_line-individual} applied to $P'$ allows us to choose $Q''$ satisfying:
    \begin{equation*}
        r_{i,i+1}\le g_{i,i+1}(Q'')\le M r_{i,i+1}/(M-1),\ i\in[n-1].
    \end{equation*}
    By the above formula $q'_i=(M R_i-q_i)/(M-1)$, we substitute $q''_i$ for $q'_i$ and obtain $q_i=M R_i-(M-1)q''_i$.
    The resulting gap $g_{i,i+1}(Q)$ equals $M r_{i,i+1}-(M-1)g_{i,i+1}(Q'')$ and the bounds on $Q''$ imply that $0\le g_{i,i+1}(Q)\le r_{i,i+1}$.
    Thus $Q$ is an order-preserving feasible solution of {\connectivity}.
    Moreover, $d(p_i,q_i)=(M-1)d(p'_i,q''_i)$ for every $i\in[n]$.
    Consequently, $\OPTconn(P)\le(M-1)\OPTspread(P')$.
    Combining both inequalities, we obtain $\OPTconn(P)=(M-1)\OPTspread(P')$.

    For the reverse transformation, we use the same argument.
    Now let the {\spreading} instance have initial positions $P$.
    By equation~\eqref{eq:bounded_gap_line-individual}, we may choose an optimal solution $Q$ satisfying $r_{i,i+1}\le g_{i,i+1}(Q)\le (M-1)r_{i,i+1}$ for every $i\in[n-1]$.
    Under the same transformation, the gaps of $Q'$ are nonnegative and at most $r_{i,i+1}$, so $Q'$ is an order-preserving feasible solution of {\connectivity}.
    In the other direction, transforming any feasible {\connectivity} solution $Q''$ for $P'$ into $Q$ by setting $q_i=M R_i-(M-1)q''_i$ yields a feasible solution for $P$.
    The same movement scaling in both directions gives $\OPTspread(P)=(M-1)\OPTconn(P')$.
\end{proof}

\section{Conclusions}
\label{sec:open_problems}

In this paper, we give linear-time reductions between {\connectivity} and {\spreading} on a line and a closed cycle.
As a consequence, the linear and cyclic versions of min-sum {\connectivity} are solvable in $O(n\log n)$ time.
We also extend the reduction to points on a line with individual thresholds when their initial order is preserved.
Several questions, however, remain open:
\begin{itemize}
    \item Does the relation between the two problems extend to points on a {\em closed cycle} with individual thresholds when their initial cyclic order is preserved?
    \item Can the reduction for individual thresholds be extended to different, partially prescribed, or unrestricted final orders?
    \item If we instead require the first and last of every $k+1$ consecutive points to be at distance at most $r$ ($k$-connectivity), how is this problem related to {\spreading}?
    \item Is there any relation between the two problems in higher dimensions?
\end{itemize}

\bibliography{bibliography}

@Article{Dumitrescu2011a,
  author    = {Dumitrescu, Adrian and Jiang, Minghui},
  journal   = {Discrete Appl. Math.},
  title     = {Constrained k-center and movement to independence},
  year      = {2011},
  doi       = {10.1016/j.dam.2011.01.008},
  issn      = {0166-218X},
  number    = {8},
  pages     = {859--865},
  volume    = {159},
  publisher = {Elsevier BV},
}

@Article{Wang2025,
  author    = {Li, Shimin and Wang, Haitao},
  journal   = {Comp. Geom. Topol.},
  title     = {Algorithms for Minimizing the Movements of Spreading Points in Linear Domains},
  year      = {2025},
  number    = {1},
  pages     = {1:1--1:15},
  volume    = {4},
  copyright = {Creative Commons Attribution 4.0 International},
  doi       = {10.57717/CGT.V4I1.21},
}

@InProceedings{Li2019A,
  author    = {Li, Shimin and Yan, Zhongjiang and Zhang, Jingru},
  booktitle = {Proc. 31st Can. Conf. Comp. Geom. ({CCCG})},
  title     = {An Optimal Algorithm for Maintaining Connectivity of Wireless Network on a Line},
  year      = {2019},
  address   = {Edmonton, Alberta, Canada},
  pages     = {78--84},
}

@InProceedings{Li2026,
  author    = {Li, Shimin},
  booktitle = {Proc. 38th Can. Conf. Comp. Geom. ({CCCG})},
  title     = {Algorithms for Connectivity Maintenance and Barrier Coverage on a Closed Cycle},
  year      = {2026},
  address   = {Orillia, Ontario, Canada},
  pages     = {192--201},
  doi       = {10.48550/ARXIV.2608.01307},
}

@InProceedings{Ghadiri2016,
  author    = {Ghadiri, Mehrdad and Yazdanbod, Sina},
  booktitle = {Proc. 28th Can. Conf. Comp. Geom. ({CCCG})},
  title     = {Minimizing the Total Movement for Movement to Independence Problem on a Line},
  year      = {2016},
  address   = {Vancouver, British Columbia, Canada},
  pages     = {15--20},
  doi       = {10.48550/ARXIV.1606.09596},
}

@InProceedings{HonoratoDroguett2026,
  author    = {Honorato-Droguett, Nicol{\'a}s and Kurita, Kazuhiro and Hanaka, Tesshu and Ono, Hirotaka and Wolff, Alexander},
  booktitle = {Proc. 20th Int. Conf. and Workshops on Algorithms and Comp. ({WALCOM} 2026)},
  title     = {Further Results on Rendering Geometric Intersection Graphs Sparse by Dispersion},
  doi       = {10.1007/978-981-95-7127-7_30},
  url       = {https://doi.org/10.48550/arXiv.2509.20903},
  pages     = {451--466},
  publisher = {Springer},
  series    = {{LNCS}},
  volume    = {16444},
  year      = {2026},
}

@Article{Demaine2009,
  author    = {Demaine, Erik D. and Hajiaghayi, MohammadTaghi and Mahini, Hamid and Sayedi-Roshkhar, Amin S. and Oveisgharan, Shayan and Zadimoghaddam, Morteza},
  journal   = {ACM Trans. Algorithms},
  title     = {Minimizing movement},
  year      = {2009},
  doi       = {10.1145/1541885.1541891},
  issn      = {1549-6333},
  number    = {3},
  pages     = {30:1--30:30},
  volume    = {5},
  publisher = {Association for Computing Machinery (ACM)},
}

@InProceedings{Czyzowicz2010,
  author    = {Czyzowicz, Jurek and Kranakis, Evangelos and Krizanc, Danny and Lambadaris, Ioannis and Narayanan, Lata and Opatrny, Jaroslav and Stacho, Ladislav and Urrutia, Jorge and Yazdani, Mohammadreza},
  booktitle = {Proc. 9th Int. Conf. Ad-Hoc, Mobile and Wireless Netw. ({ADHOC-NOW} 2010)},
  title     = {On Minimizing the Sum of Sensor Movements for Barrier Coverage of a Line Segment},
  doi       = {10.1007/978-3-642-14785-2_3},
  pages     = {29--42},
  publisher = {Springer Berlin Heidelberg},
  series    = {{LNCS}},
  volume    = {6288},
  year      = {2010},
}

@Article{Andrews2016,
  author       = {Andrews, Aaron M. and Wang, Haitao},
  title        = {Minimizing the Aggregate Movements for Interval Coverage},
  doi          = {10.1007/s00453-016-0153-8},
  issn         = {1432-0541},
  number       = {1},
  pages        = {47--85},
  volume       = {78},
  year         = {2017},
  journal      = {Algorithmica},  publisher    = {Springer Science and Business Media LLC},
}

@Article{Kapelko2022,
  author    = {Kapelko, Rafał},
  journal   = {Sensors},
  title     = {Analysis of the Threshold for Energy Consumption in Displacement of Random Sensors},
  year      = {2022},
  issn      = {1424-8220},
  number    = {22},
  pages     = {8789},
  volume    = {22},
  doi       = {10.3390/s22228789},
  publisher = {MDPI AG},
}

@Article{Kranakis2016,
  author   = {Kranakis, Evangelos and Shaikhet, Gennady},
  journal  = {J. Appl. Probab.},
  title    = {Sensor Allocation Problems on the Real Line},
  year     = {2016},
  doi      = {10.1017/jpr.2016.33},
  number   = {3},
  pages    = {667--687},
  volume   = {53},
  comment  = {Interference and Coverage Problems},
}

@InProceedings{Bredin2005,
  author     = {Bredin, Jonathan L. and Demaine, Erik D. and Hajiaghayi, MohammadTaghi and Rus, Daniela},
  booktitle  = {Proc. 6th ACM Int. Symp. Mobile Ad Hoc Netw. Comp. ({MobiHoc} 2005)},
  title      = {Deploying sensor networks with guaranteed capacity and fault tolerance},
  year       = {2005},
  pages      = {309--319},
  publisher  = {ACM},
  doi        = {10.1145/1062689.1062729},
}

@Article{Zavlanos2007,
  author    = {Zavlanos, Michael M. and Pappas, George J.},
  journal   = {IEEE Trans. Robot.},
  title     = {Potential Fields for Maintaining Connectivity of Mobile Networks},
  year      = {2007},
  issn      = {1552-3098},
  number    = {4},
  pages     = {812--816},
  volume    = {23},
  doi       = {10.1109/tro.2007.900642},
  publisher = {Institute of Electrical and Electronics Engineers (IEEE)},
}

@Article{Wang2006,
  author    = {Wang, G. and Cao, G. and La Porta, T.F.},
  journal   = {IEEE Trans. Mobile Comp.},
  title     = {Movement-assisted sensor deployment},
  year      = {2006},
  issn      = {1536-1233},
  number    = {6},
  pages     = {640--652},
  volume    = {5},
  doi       = {10.1109/tmc.2006.80},
  publisher = {Institute of Electrical and Electronics Engineers (IEEE)},
}

@Article{Li2019,
  author  = {Li, Shimin and Wang, Haitao},
  journal = {J. Comp. Geom.},
  title   = {Separating overlapped intervals on a line},
  year    = {2019},
  doi     = {10.20382/JOCG.V10I1A11},
  number  = {1},
  pages   = {281--321},
  volume  = {10},
}

@InProceedings{Mehrandish2011,
  author    = {Mehrandish, Mona and Narayanan, Lata and Opatrny, Jaroslav},
  booktitle = {Proc. {IEEE} Wireless Commun. Netw. Conf. ({WCNC} 2011)},
  title     = {Minimizing the number of sensors moved on line barriers},
  year      = {2011},
  doi       = {10.1109/WCNC.2011.5779210},
  pages     = {653--658},
  publisher = {IEEE},
}

@InProceedings{Gaspers2017,
  author    = {Gaspers, Serge and Gudmundsson, Joachim and Mestre, Juli{\'a}n and R{\"u}mmele, Stefan},
  booktitle = {Proc. 28th Int. Symp. Algorithms and Comp. ({ISAAC} 2017)},
  title     = {Barrier Coverage with Non-uniform Lengths to Minimize Aggregate Movements},
  year      = {2017},
  doi       = {10.4230/LIPIcs.ISAAC.2017.37},
  pages     = {37:1--37:13},
  publisher = {Schloss Dagstuhl -- Leibniz-Zentrum f{\"u}r Informatik},
  series    = {{LIPIcs}},
  volume    = {92},
}

@Article{Chen2013,
  author    = {Chen, Danny Z. and Gu, Yan and Li, Jian and Wang, Haitao},
  journal   = {Discrete Comp. Geom.},
  title     = {Algorithms on Minimizing the Maximum Sensor Movement for Barrier Coverage of a Linear Domain},
  year      = {2013},
  doi       = {10.1007/s00454-013-9525-x},
  number    = {2},
  pages     = {374--408},
  volume    = {50},
  publisher = {Springer},
}

@InProceedings{Czyzowicz2009,
  author    = {Czyzowicz, Jurek and Kranakis, Evangelos and Krizanc, Danny and Lambadaris, Ioannis and Narayanan, Lata and Opatrny, Jaroslav and Stacho, Ladislav and Urrutia, Jorge and Yazdani, Mohammadreza},
  booktitle = {Proc. 8th Int. Conf. Ad-Hoc, Mobile and Wireless Netw. ({ADHOC-NOW} 2009)},
  title     = {On Minimizing the Maximum Sensor Movement for Barrier Coverage of a Line Segment},
  year      = {2009},
  doi       = {10.1007/978-3-642-04383-3_15},
  pages     = {194--212},
  publisher = {Springer},
  series    = {{LNCS}},
  volume    = {5793},
}

\end{document}